\documentclass[11pt]{article}

\usepackage{latexsym}
\usepackage{amsmath}
\usepackage{amssymb}
\usepackage{amsfonts}
\usepackage{bm}
\usepackage{bbm}
\usepackage{mathrsfs}
\usepackage[margin=0.75in]{geometry}
\usepackage{amsthm}
\usepackage{autobreak}
\usepackage{verbatim}
\usepackage{graphicx}
\usepackage{epsfig}
\usepackage{url}
\usepackage{float}
\usepackage{color}
\usepackage{fullpage}
\usepackage{appendix}
\usepackage[nomessages]{fp}
\usepackage{enumerate}
\usepackage[framemethod=tikz]{mdframed}

\usepackage{multirow}
\usepackage{diagbox}

\newtheorem{theorem}{Theorem}[section]

\newtheorem{proposition}[theorem]{Proposition}

\newtheorem{lemma}[theorem]{Lemma}

\newtheorem{remark}[theorem]{Remark}
\theoremstyle{definition}
\newtheorem{definition}[theorem]{Definition}

  \def\01{\{0,1\}}

    \def\01{\{0,1\}}

    \usepackage[noline, boxed]{algorithm2e}

    \newcommand {\br} [1] {\ensuremath{ \left( #1 \right) }}

    \newcommand {\abs} [1] {\ensuremath{ \left| #1 \right| }}
    
    \newcommand {\set} [1] {\ensuremath{ \left\lbrace #1 \right\rbrace }}

\begin{document}

\title{Complexity of Eccentricities and All-Pairs Shortest Paths\\in the Quantum CONGEST Model}

\author{
Changsheng Wang \thanks{State Key Laboratory for Novel Software Technology, Nanjing University, email: wangcs@smail.nju.edu.cn}
\and
Xudong Wu \thanks{State Key Laboratory for Novel Software Technology, Nanjing University, email: xdwu@smail.nju.edu.cn}
\and
Penghui Yao \thanks{State Key Laboratory for Novel Software Technology, Nanjing University, email: pyao@nju.edu.cn}
}

\maketitle

\begin{abstract}
Computing the distance parameters of a network, including the diameter, radius, eccentricities and the all-pairs shortest paths (APSP) is a central problem in distributed computing. This paper investigates the distance parameters in the quantum CONGEST models and establishes almost linear lower bounds on eccentricities and APSP, which match the classical upper bounds. Our results imply that there is not quantum speedup for these two problems. In contrast with the diameter and radius, exchanging quantum messages is able to save the communication when the networks have low diameters~\cite{GallM18,magniez2021quantum}. We obtain the lower bounds via a reduction from the two-way quantum communication complexity of the set intersection~\cite{2003QuantumSP}.
\end{abstract}

\newpage


\section{Introduction}

In distributed computing, the graph $G=(V,E)$ represents the topology of the network, each node represents a processor with an unique ID, and messages can be transferred directly through the edges connecting two nodes.
In the classical CONGEST model, the computation proceeds with round-based synchrony and each node can send one $O(\log n)$-bit message to each adjacent node per round, where $n$ denotes the number of nodes.
More specifically, protocols in the model are executed round by round.
In each round, each node can send/receive a message of $O(\log n)$ bits to/from each of its neighbours, i.e., each pair of neighbours can exchange a message of $O(\log n)$ bits.
After that, all nodes implement local computation and enter the next round simultaneously.
As for the quantum version defined in \cite{elkin2014can}, the only difference is that each node can send  qubits instead of classical bits.
The \textit{round complexity} of a distributed problem $f$ in the CONGEST model (both classical and quantum version) denotes the smallest number of communication rounds needed to compute $f$.

The diameter of a graph, denoted by $D$, is defined as the maximum distance between any two nodes.
It is a fundamental parameter of the network since in at least $D$ rounds, every two nodes are possible to know each other.
The eccentricity with respect to some node $v$ is defined as the maximum distances from $v$ to any other node.
So the diameter is the maximum eccentricities over all nodes.
The radius of a graph is defined to be the minimum eccentricities over all nodes.
Diameters, radius and eccentricities are three basic distance parameters in distributed network.

We study the distance computation in the quantum CONGEST model and focus on investigating whether introducing quantum communication can speed up distributed computation compared to the classical setting.
\\[3pt]
\noindent
\textbf{Related works.}
Elkin, Klauck, Nanongkai and Pandurangan \cite{elkin2014can} proved that quantum communication does not offer significant advantages over classical communication for many fundamental problems in distributed computing, including minimum spanning tree (MST), minimum cut, $s$-source distance, shortest path tree (SPT) and shortest $s$-$t$ paths.

However, based on distributed Grover's search~\cite{Grover96}, Le Gall and Magniez~\cite{GallM18} proposed a $\tilde O(\sqrt{nD})$-round quantum algorithm to exactly compute the diameter in the CONGEST model when the graph has a low diameter, where $D$ is the diameter of the underlying network graph, exhibiting a quantum advantage in the CONGEST model, as the classical round complexity is $\Omega\br{n}$.
They also proved a tight lower bound $\tilde\Omega(\sqrt{nD})$ for any distributed quantum algorithms if each node can use only poly$(\log n)$ quantum bits. An unconditional quantum lower bound $\tilde\Omega(\sqrt{n})$ for diameters was also proved in \cite{GallM18}, which was later improved to $\Omega(\sqrt[3]{nD^2}+\sqrt{n})$ by Magniez and Nayak~\cite{magniez2021quantum}.

In the classical CONGEST model, there is a series of studies on distance computation.
Earlier, Frischknecht, Holzer and Wattenhofer~\cite{frischknecht2012networks} showed that computing the diameter of an unweighted graph in CONGEST model requires $\tilde\Omega(n)$ rounds, which is tight up to a logarithmic factors since even the All-Pairs Shortest Path problem(APSP) on an unweighted graph can be resolved in $O(n)$ rounds \cite{peleg2012distributed,holzer2012optimal}.
Abboud, Censor-Hillel and Khoury~\cite{AbboudCK16} later gave a same lower bound of $\tilde\Omega(n)$ on diameters even in sparse networks.
As a result, the round complexities of computing diameter, radius eccentricities and APSP in the classical CONGEST model are all in $\tilde\Theta(n)$ regime.

Regarding the approximate distance computation, an $\tilde\Omega(n)$ lower bound holds even for $(3/2-\epsilon)$-approximate diameter, $(3/2-\epsilon)$-approximate radius, and $(5/3-\epsilon)$-approximate eccentricities in \cite{AbboudCK16}, where $\epsilon>0$ is a constant. For larger approximation ratios,  Holzer , Peleg, Roditty and Wattenhofer~\cite{HolzerPRW14} and Ancona et al.~\cite{ancona2020distributed} present $\tilde O(\sqrt{n}+D)$-round distributed algorithms for $3/2$-approximate diameter, $3/2$-approximate radius, and $5/3$-approximate eccentricities.


\noindent
\textbf{Main results.}
Le Gall and Magniez~\cite{GallM18} showed that 
quantum computation can speed up computing diameter (and radius by similar argument). We focus on the round complexity of eccentricities and APSP , completing the picture on the complexity of distance computation in quantum CONGEST models.
\begin{theorem}
    \label{apsp_bound}
    Any quantum protocol computing all-pairs shortest paths in the CONGEST model requires $\Omega(\frac n{\log n})$ rounds.
\end{theorem}
\begin{theorem}
    \label{ecc_bound}
    Any quantum protocol computing the exact eccentricities of all nodes in the CONGEST model requires $\Omega(\frac n{\log^2n})$ rounds.
\end{theorem}
\begin{theorem}
    \label{ecc_approx_bound}
    For all constant $0<\varepsilon<\frac23$, any quantum protocol computing a $(\frac53-\varepsilon)$-approximation of all eccentricities in the CONGEST model requires $\Omega(\frac n{\log^3n})$ rounds.
\end{theorem}
\noindent
\begin{remark}
The lower bounds on 
eccentricities and APSP 
we establish still hold even if the networks have constant diameters. Thus there is no quantum advantages even for the networks with constant diameters. However, for such networks, it is known that the quantum round complexities of the diameter and radius are quadratically smaller than the classical round complexities in the CONGEST model~\cite{frischknecht2012networks,GallM18,magniez2021quantum}.

\end{remark}
Our lower bounds are tight up to a logarithmic factor. Thus quantum computation cannot speedup the eccentricities and APSP computation substantially.
Our results and related works are collected in Table~\ref{tab:summary}.
\linespread{1.8}
\begin{table}[ht]
    \centering
    \scriptsize
    \begin{tabular}{|c|c|c|c|c|c|}
        \hline
        \multirow{2}*{\textbf{Problem}} & \multirow{2}*{\textbf{Approx.}} & \multicolumn{2}{c|}{\textbf{Upper Bound $\tilde O(\cdot)$}} & \multicolumn{2}{c|}{\textbf{Lower Bound $\tilde\Omega(\cdot)$}} \\
        \cline{3-6}
        ~ & ~ & \textbf{Classical} & \textbf{Quantum} & \textbf{Classical} & \textbf{Quantum} \\
        
        \hline
        APSP & Exact & $n$~\cite{peleg2012distributed,holzer2012optimal} & $n$ & $n$ & $n$ (Theorem~\ref{apsp_bound}) \\
        
        \hline
        \multirow{3}*{Eccentricities} & Exact & \multicolumn{2}{c|}{\multirow{2}*{$n$}} & $n$ & $n$ (Theorem~\ref{ecc_bound}) \\
        \cline{2-2} \cline{5-6}
        ~ & $5/3-\varepsilon$ & \multicolumn{2}{c|}{~} & $n$~\cite{AbboudCK16} & $n$ (Theorem~\ref{ecc_approx_bound}) \\
        \cline{2-6}
        ~ & $5/3$ & $\sqrt{n}+D$~\cite{ancona2020distributed} & & & \\
        
        \hline
        \multirow{3}*{Diameter} & Exact & \multirow{2}*{$n$} & \multirow{2}*{$\sqrt{nD}$~\cite{GallM18}} & $n$~\cite{frischknecht2012networks} & $\sqrt[3]{nD^2}+\sqrt{n}$~\cite{magniez2021quantum} \\
        \cline{2-2} \cline{5-6}
        ~ & $3/2-\varepsilon$ & ~ & ~ & $n$~\cite{AbboudCK16} & \\
        \cline{2-6}
        ~ & $3/2$ & $\sqrt{n}+D$~\cite{HolzerPRW14} & $\sqrt[3]{nD}+D$~\cite{GallM18} & & \\
        
        \hline
        \multirow{3}*{Radius} & Exact & \multirow{2}*{$n$} & \multirow{2}*{$\sqrt{nD}$} & \multirow{2}*{$n$~\cite{AbboudCK16}}& $\sqrt[3]{nD^2}+\sqrt{n}$ (Theorem~\ref{radius})\\
        \cline{2-2} \cline{6-6}
        ~ & $3/2-\varepsilon$ & ~ & ~ & ~ & \\
        \cline{2-6}
        ~ & $3/2$ & $\sqrt{n}+D$~\cite{ancona2020distributed} & & & \\
        
        \hline
    \end{tabular}
    
    \normalsize
    \linespread{1}
    \caption{Upper and Lower bounds of distance computation in the classical and quantum CONGEST model. $\varepsilon>0$ is a constant.}
    \label{tab:summary}
\end{table}
\linespread{1}

\noindent
\textbf{Techniques overview.}
A well-known technique to prove lower bounds on CONGEST models is via two-party communication complexity.
The best lower bounds on the classical round complexity of computing the distances are induced from the two-party communication complexity of the disjointness function $\text{DISJ}_k$.

The lower bounds on the number of rounds needed to compute the diameter/radius exactly and approximately in quantum CONGEST model can be obtained by simply using the the quantum communication complexity of $\text{DISJ}_k$.
However, it seems difficult to obtain tight lower bounds on the round complexities eccentricities and APSP via a reduction from disjointness.

We construct a reduction from the intersection function $ \text{INT}_{k,\rho} $ to APSP, where the $ \text{INT}_{k,\rho} $ is a Boolean function that counts the number of intersecting elements and  disjointness function is a special case of $ \text{INT}_{k,\rho} $ when $\rho=1$.
Razborov~\cite{2003QuantumSP} showed a tight bound on the quantum communication complexity of the $ \text{INT}_{k,\rho} $.
Assuming that there exists a quantum distributed algorithm that computes APSP, we can use the algorithm to construct a protocol that computes the intersection function, which implies a lower bound on APSP.
To the best of our knowledge, this is the first time that the intersection function is used to prove the lower bound on quantum distributed computation.
Analogously, we are able to prove a tight bound of eccentricities via the quantum communication complexity of the intersection function.

\section{Preliminaries}

\subsection{Classical and quantum CONGEST model}

In the classical CONGEST model, the communication network can be seen as a graph $G=(V,E)$.
Usually we assume that there are $n$ nodes and $m$ edges in $G$, and the nodes are assigned with unique identifiers.
Each node represents a processor with unlimited computational power, i.e., the consumption of any local computation in a single processor is ignored.
Each edge connecting two nodes represents a communication channel with $O(\log n)$ bits of bandwidth.

In the quantum CONGEST model, adjacent nodes are allowed to exchange \textit{quantum bits (qubits)}, i.e., the classical channels are replaced by quantum channels with the same bandwidth $O(\log n)$.
And naturally each node can locally do some quantum computation.
Nodes may own qubits which are entangled with the qubits owned by other nodes.
In this paper, we assume that initially distinct nodes do not share entanglement.
But they can, for example, locally create a pair of entangled qubits, and send one to other nodes.

For both classical and quantum CONGEST models,
the algorithm is implemented round by round in a synchronous manner.
In each round, every node sends/receives a message of $O(\log n)$ bits to/from each neighbour, then implements local computation.
At the end of the algorithm, every node has its own output (maybe empty).
The round complexity of an algorithm is defined to be the number of communication rounds in the process %
in the worst case.
And the round complexity of a problem is the least round complexity of any algorithm solving this problem.


\subsection{Graph notation and problem definition}

For a graph $G=(V,E)$, the distance between two nodes $u$ and $v$ is the length of the shortest path between $u$ and $v$.
We use $d(u,v)$ to denote the distance between $u$ and $v$.
The eccentricity with respect to some node $u$ is denoted by
$$e(u)=\max_{v\in V}d(u,v).$$
The diameter of a graph $G$, denoted by 
$D$, is the maximum eccentricities over all nodes, i.e., the maximum distance between any two nodes; while the radius of a graph $G$, denoted by $R$, is the minimum eccentricities over all nodes:
$$D=\max_{u\in V}e(u)=\max_{u,v\in V}d(u,v)\text{ and }R=\min_{u\in V}e(u).$$
Any node with the minimum eccentricity is called a \textit{center} of the graph, so the radius is the eccentricity of any center node.

Before describing problems of distance computation in the distributed setting, recall the definition of approximation for optimization problems.

Given a maximization problem $\mathcal P$ and an instance $x$,  $OPT(x)$ denotes the value of the optimal solution of the instance $x$,  and $SOL_{\mathcal A}(x)$ denotes the solution that the algorithm $\mathcal A$ obtains for instance $x$. For any $\rho\geq 1$, we say that the algorithm $\mathcal A$ computes a $\rho$-(multiplicative) approximation  of $\mathcal P$ if
$\frac1\rho\cdot OPT(x)\le SOL_{\mathcal A}(x)\le OPT(x)$
for any instance $x$.
Similarly for a minimization problem $\mathcal P$ and $\rho\ge1$, we say that an algorithm $\mathcal A$ computes a $\rho$-approximation of $\mathcal P$ if
$OPT(x)\le SOL_{\mathcal A}(x)\le\rho\cdot OPT(x)$
for any instance $x$.

\begin{definition}[Diameter/radius]
    Given a network with underlying graph $G=(V,E)$, the diameter computation, which is a maximization problem, requires that all nodes should have the same output, which is the exact value (or an approximation) of the diameter of $G$.
    Similarly the radius computation, which is a minimization problem, requires nodes to output the same value, which is the exact value (or an approximation) of the radius of $G$.
\end{definition}

\begin{definition}[Eccentricities]
    Given a network with the underlying graph $G=(V,E)$, the eccentricities computation requires each node $u$ to output its eccentricity, or an approximation.
\end{definition}
\noindent
\textbf{Remarks.}
If an algorithm $\mathcal A$ has node $u$ outputs $\hat e(u)$, we say that, for $\rho\ge1$, $\mathcal A$ computes a $\rho$-approximation of the eccentricity if
$$\forall\ u\in V,\ \frac1\rho\cdot e(u)\le \hat e(u)\le e(u).$$

\begin{definition}[APSP]
    Given a network with the underlying graph $G=(V,E)$, the APSP computation requires each node $u$ to output the distance between $u$ and $v$ for each node $v$.
\end{definition}


\subsection{Two-party communication complexity}

\label{sec:cc}

A main tool to obtain lower bounds in the CONGEST model is via reductions to two-way communication complexity and apply the various existing methods proving lower bounds on communication complexity.

Communication complexity was first introduced by Yao in~\cite{yao1979some}. Consider two players, usually called Alice and Bob, and assume that Alice and Bob receives inputs $x, y\in\{0,1\}^k$, respectively. The players want to compute a Boolean function $f:\{0,1\}^k\times\{0,1\}^k\to\{0,1\}$ by exchanging messages. At the end of the protocol, both Alice and Bob output $f(x,y)$.  The minimum communication required for the players to compute the function is the communication complexity of $f$.

In  \cite{yao1993quantum}, Yao introduced the model of quantum communication complexity, where the players are allowed to communicate with qubits. The communication cost of a quantum protocol is the maximum (over all inputs) number of qubits that the players exchange. The quantum communication complexity of a function $f$ is the minimum communication cost of any quantum protocol that computes $f$ with probability at least $2/3$.

\begin{definition}[Disjointness]
    For any integer $k\ge1$, the disjointness function $\text{DISJ}_k:\{0,1\}^k\times\{0,1\}^k\to\{0,1\}$ is the function such that $\text{DISJ}_k(x,y)=0$ if and only if there exists an index $p\in[k]$ with $x_p=y_p=1$, where $[k]=\{1,\cdots,k\}$.
\end{definition}
It is well known that the (randomized) classical communication complexity of the disjointness function is $\Theta(k)$~\cite{10.1016/0304-3975(92)90260-M}, while its quantum communication complexity is $\Theta(\sqrt{k})$~\cite{1238194,2003QuantumSP}. Braverman, Garg, Ko, Mao and Touchette further proved an almost tight lower bound on the bounded-round quantum communication complexity of disjointness.
\begin{lemma}[\cite{braverman2018near}]
    \label{disj_bound}
    The $r$-round quantum communication complexity of disjointness $\textup{DISJ}_k$ is $\Omega\left(\frac k{r\log^8r}+r\right)$.
\end{lemma}

The classical lower bound of the disjointness has been widely applied to obtain numerous tight lower bounds in the classical CONGEST model. However, its quantum lower bound seems to not always have enough capability to capture the hardness in the quantum CONGEST model. We further introduce the intersection function and its quantum communication complexity.
\begin{definition}[Intersection]\label{def:int}
    For an integer $k\ge1$ and $1\le\rho\le k$, the intersection function $\text{INT}_{k,\rho}:\{0,1\}^k\times\{0,1\}^k\to\{0,1\}$ is the function such that $\text{INT}_{k,\rho}(x,y)=0$ if and only if the number of indices $p\in[k]$ with $x_p=y_p=1$ is less than $\rho$.
\end{definition}

\begin{lemma}[\cite{2003QuantumSP}]
    \label{int_bound}
    The quantum communication complexity of the intersection function $\textup{INT}_{k,\rho}$ is $\Omega(\sqrt{k\rho})$ if $\rho\le k/2$, and $\Omega(k-\rho)$ if $\rho>k/2$.
\end{lemma}

\section{Warm up}

Based on Grover's search for the maximum eccentricities over all nodes,
Le Gall and Magniez~\cite{GallM18} proposed a distributed quantum algorithm of computing the exact diameter within $ \Tilde{O}(\sqrt{nD}) $ rounds in the CONGEST model, where $ D $ is the diameter of the network.
Similarly, apply Grover's search for the minimum eccentricities over all nodes, we can obtain an $ \Tilde{O}(\sqrt{nD}) $-round distributed algorithm to compute the exact radius of the network. 

The diameter is the maximum eccentricities over all nodes and the radius is the minimum eccentricities over all nodes. These two parameters have many similar properties. Before we introduce the lower bound of radius computation, we should know about how to prove the lower bound of the diameter computation.

Let $\mathcal{G}_d$ be a line network, which consists of $d+1$ nodes, denoted by $A_0,\ldots, A_d$. Here $A_0$ and $A_d$ are the two end points and $A_1,\ldots, A_{d-1}$ are the intermediate nodes. Each edge of $ \mathcal{G}_d $ is a quantum channel of bandwidth $B$ qubits.
The nodes $A_0$ and $A_d$ receive $k$-bit inputs $x, y\in \{0,1\}^k $, respectively. The disjointness on a line $L_{k,d}$ is a computational problem over the network $ \mathcal{G}_d $ that requires to compute the $\text{DISJ}_k(x,y)$. We say an algorithm solves $L_{k,d}$ with probability $ p $, if  there exists a node outputs the right answer with probability at least $p$.
The following lemma relates the round complexity of $L_{k,d}$ to the two-party communication complexity of $\text{DISJ}_k(x,y)$.

\begin{lemma}[\cite{GallM18}]
	\label{simulationarg}
	Let $ d $ and $ r $ be any positive integers. If there exists an $ r $-round quantum distributed algorithm, in which each intermediate node uses at most $ s $ qubits of memory, that computes function $L_{k,d}$ over $ {\mathcal{G}_d} $ with probability $ p $, then there exists a $ O(r/d) $-round two-party quantum protocol computing $\text{DISJ}_k$ with probability $ p $ using $ O(r(B + s)) $ qubits of communication, where bandwidth $B=\Theta(\log n)$ in the CONGEST model.
\end{lemma}

Le Gall and Magniez~\cite{GallM18} showed that any quantum distributed algorithm that computing the diameter of the network requires $ \Tilde{\Omega}(\sqrt{nD/s}) $ rounds if each node uses at most $s$ qubits of memory where $D$ is the diameter of the network via a reduction from disjointness on a line $L_{k,d}$

In~\cite{AbboudCK16}, the authors also exhibited a  reduction from disjointness to diameter computation \cite{AbboudCK16}. Thus Lemma~\ref{disj_bound} implies that computing the diameter need $ \Tilde{\Omega}(\sqrt{n}) $ rounds without memory restriction.

Magniez and Nayak~\cite{magniez2021quantum} proved a new  lower bound of the disjointness on a line (lemma \ref{setDisjonaline}) and established a new lower bound $ \Tilde{\Omega}(\sqrt[3]{nD^2}+\sqrt{n})$ for the diameter in the CONGEST model without memory restriction.

\begin{lemma}[\cite{magniez2021quantum}]
   \label{setDisjonaline}
    Any quantum communication protocol with error probability at most 1/3 for the disjointness on the line $L_{k,d}$ requires $ {\Omega}(\sqrt[3]{kD^2/B}) $ rounds, where $B$ is the the bandwidth of each communication channel.
\end{lemma}

Based on the approach of proving the lower bound of computing the diameter, we can directly obtain the the following theorem.

\begin{theorem}
    \label{radius}
	The number of rounds needed for any quantum protocol to compute the radius of a sparse network in the CONGEST model is $ \Tilde{\Omega}(\sqrt{n}+\sqrt[3]{n D^2}) $, where $ D $ denotes the diameter of the network. If each node uses at most $ s $ qubits of memory, computing the radius of the network requires $ \tilde{\Omega}(\sqrt{nD/s}) $ rounds.
\end{theorem}

\begin{proof}[Proof]

First of all, there exists a graph $G$ reducing $\text{DISJ}_n(x,y)$ to radius computation \cite{AbboudCK16}, where the graph $G$ has $ \Theta(n) $ nodes. The graph $G$ has two parts, one is simulated by Alice (denoted by $G_a$) and the other is simulated by Bob (denoted by $G_b$), which has $\Theta(\log n)$ edges connecting $G_a$ and $G_b$.
 Assume that there exists a $r$-round quantum distributed algorithm $\mathcal{A}$ that computes the radius, we can use the algorithm to construct a protocol that computes the $ \text{DISJ}_n(x,y)$, where $x\in\{0,1\}^n$ and $y\in\{0,1\}^n$ are Alice's input and Bob's input, respectively. 
 
 The protocol works as follows. Alice and Bob can jointly simulate the graph $G$: Alice simulates
 $G_a$(which depends on $x$), while Bob simulates
 $G_b$(which depends on $y$). To simulate the $r$-round quantum distributed algorithm $\mathcal{A}$, Alice and Bob need to exchange messages corresponding to the communication occurring along the $\Theta(\log n)$ edges between $G_a$ and $G_b$. Because there are $\Theta(\log n)$ channels(edges) and the bandwidth of each channel is $O(\log n)$, one round of communication in algorithm $\mathcal{A}$ can be replaced with two messages of $O(\log^2 n)$ qubits. Finally, Alice and Bob can compute the radius of $G$, thus compute $ \text{DISJ}_n(x,y)$.
 From lemma \ref{disj_bound}, we conclude that $r\log^2 n=\Tilde{\Omega}(\frac{n}{r}+r) $, which implies $ r=\Tilde{\Omega}(\sqrt{n}) $.
 
Replacing each edges connecting $G_a$ and $G_b$ with a path of length $d$, we obtain an instance for radius computation, which is reduced from $L_{n,d}$ with $d=\Theta(D)$. 

We consider the following two cases.

\begin{itemize}
    \item {\em There is no upper bound on the memory of each node.}   Assume that there exists an $r$-round quantum distributed algorithm that computes the radius, we can obtain an $r$-round protocol for $L_{k,d}$. By Lemma~ \ref{setDisjonaline}, we  obtain an lower bound $\Tilde{\Omega}(\sqrt[3]{nD^2}) $ as desired.

    \item {\em Each node uses at most $s$ qubits.}  By Lemma~\ref{simulationarg}, there exists an $ O(r/d) $-round quantum communication protocol for  $ \text{DISJ}_n(x,y)$ with the communication cost  $O(r(\log^2 n+s))$. By Lemma~\ref{disj_bound}, $ r(\log^2 n+s)=\tilde{\Omega}(\frac{n}{r/d}+r/d)$, which implies $ r=\tilde{\Omega}(\sqrt{nD/s}) $.
\end{itemize}
\end{proof}

\section{Lower bounds}

We briefly outline the reductions from two-party communication complexity to the computation in the CONGEST model.
And thus the lower bound on communication complexity implies the lower bound in the CONGEST model.

Given a distributed problem $\mathcal P$, we briefly describe the reduction from communication complexity to $\mathcal P$.
For a Boolean function $f:\{0,1\}^k\times\{0,1\}^k\to\{0,1\}$, let $x, y\in\{0,1\}^k$ be the input of $f$.
The reduction is via a construction of instance $G=(V,E)$ of $\mathcal P$, which consists of three steps shown in Figure~\ref{fig:construction}.
\begin{figure}
    \begin{mdframed}
        \begin{enumerate}
            \item To construct an instance $G=(V,E)$, the vertex set $V$ are partitioned into two disjoint sets $V_a$ and $V_b$;
            
            \item Edge set $E$ includes several edges that are independent of $x$ and $y$,
            such as edges between $V_a$ and $V_b$;
            
            \item A set of edges between nodes in $V_a$ (resp.~$V_b$) is created according to $x$ (resp.~$y$), denoted by $E_a$ (resp.~$E_b$).
        \end{enumerate}
    \end{mdframed}
    \caption{Construction of reduction graph}
    \label{fig:construction}
\end{figure}

If $f(x,y)$ can be inferred from the solution of $\mathcal P$ on $G$, then the lower bound on the complexity of $\mathcal P$ can be obtained from the following argument.
1) Assume that there is a $r$-round distributed algorithm $\mathcal A$ solving $\mathcal P$;
2) the player Alice (resp.~Bob) knows $x$ (resp.~$y$) and simulate the execution of $\mathcal A$ on $G$ for each node in $V_a$ (resp.~$V_b$);
3) Alice and Bob obtain the solution of $\mathcal P$ on $G$ and compute $f(x,y)$.
The communication between Alice and Bob is exactly the same with the communication between nodes in $V_a$ and $V_b$ during running $\mathcal A$ on $G$.
Suppose there are $s$ edges between nodes in $V_a$ and $V_b$.
Recall that edges have bandwidth $O(\log n)$ in CONGEST model.
The communication complexity of $f$ is $O(rs\cdot\log n)$, which should be beyond any lower bound of the communication complexity of $f$.


\subsection{APSP}

We follow the above framework to prove a nearly tight lower bound on the round complexity of APSP in the quantum CONGEST model, i.e., to prove theorem~\ref{apsp_bound} via a reduction from the intersection functions in two-party communication setting.

Let $n\ge5$ be an integer.
Set $s=\lfloor(n-1)/2\rfloor$, $k=s(s-1)/2$ and $\rho\leq k$. Alice and Bob are 
given an instance $\br{x,y}\in\set{0,1}^k\times\set{0,1}^k$ of the intersection function. Alice and Bob who receive binary string $x,y$ of length $k$ respectively, are supposed to compute $\text{INT}_{k,\rho}\br{x,y}$. We establish the reduction by first 
constructing an instance $G=(V,E)$ of APSP. Let $\abs{V}=n$.
In the first step of Figure~\ref{fig:construction},  partition  $V=V_a\uplus V_b$ where $V_a=\{a_0,a_1,a_2,\cdots,a_s\}$ and $V_b=\{b_0,b_1,b_2,\cdots,b_s\}$ if $n$ is even, or $V_b=\{b_1,b_2,\cdots,b_s\}$ if $n$ is odd.
Next, if $n$ is even, we include all $(a_0,a_i)$, $(b_0,b_i)$ and $(a_i,b_i)$ in $E$ for $i\in[s]$. If $n$ is odd, we include all $(a_0,a_i)$ and $(a_i,b_i)$ in $E$.
The rest of edges are constructed according to the inputs $x,y\in\{0,1\}^k$. Notice that  $k=s(s-1)/2$. 
There are $k$ pairs of distinct indices in $[s]$.
Each pair $(i,j)$ with $1\le i<j\le s$ is associated with an unique index $p\in[k]$, and the existence of the edge $(a_i,a_j)$ (resp.~$(b_i,b_j)$) is determined by $x_p$ (resp.~$y_p$).
More precisely, suppose these pairs $(i,j)$ are sorted in lexicographic order, and $(i_p,j_p)$ denotes the $p^\text{th}$ pair.
For $p\in[k]$, if $x_p=0$, Alice adds an edge between $a_{i_p}$ and $a_{j_p}$ into $E$.
If $y_p=0$, Bob adds an edge $(b_{i_p},b_{j_p})$.
Formally, we have
$$
E=\bigcup_{i\in[s]}\{(a_i,b_i),(a_0,a_i),(b_0,b_i)\}
\cup\bigcup_{p\in[k]:x_p=0}\{(a_{i_p},a_{j_p})\}
\cup\bigcup_{p\in[k]:y_p=0}\{(b_{i_p},b_{j_p})\}.
$$
An example for $n=8$ (and thus $k=3$), 
$x=(0,1,0)$ and $y=(1,1,0)$
can be found in figure \ref{apsp_example}.
\begin{figure}[ht]
	\centering
	\includegraphics[width=0.8\linewidth]{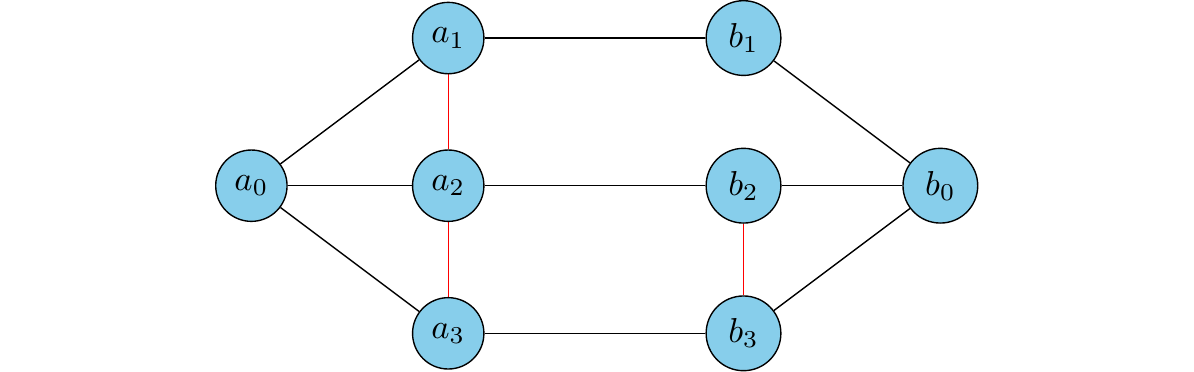}
	\caption{An example of reduction graph (APSP)}
	\label{apsp_example}
\end{figure}

\begin{proposition}
    For $1\le i<j\le s$, suppose $(i,j)$ is the $p^\text{th}$ pair in lexicographic order, i.e., $(i_p,j_p)=(i,j)$.
    If $x_p=y_p=1$, $d(a_i,b_j)=3$, otherwise $d(a_i,b_j)=2$.
\end{proposition}
\begin{proof}
    We use the graphical notation $u\to v$ to denote an edge $(u,v)$ traversed in a path.
    There is always a path of the form $(a_i\to a_0\to a_j\to b_j)$, so $d(a_i,b_j)\le3$.
    If $x_p=y_p=1$, $a_i$ and $b_j$ share no common neighbour, and thus $d(a_i,b_j)=3$.
    Otherwise they share a common neighbour $a_j$ if $x_p=0$, or $b_i$ if $y_p=0$, and thus $d(a_i,b_j)=2$.
\end{proof}

So the number of $(i,j)$ such that $1\le i<j\le s$ and $d(a_i,b_j)=3$ is the same with the number of $p\in[k]$ with $x_p=y_p=1$, or the size of the intersection of $x$ and $y$.
$|x\cap y|$ can be inferred from the solution of APSP, and so does $\text{INT}_{k,\rho}(x,y)$ for any $\rho$.
We are now ready to prove the quantum lower bound of APSP.
\begin{proof}[Proof of Theorem~\ref{apsp_bound}]
    Let $\mathcal A$ be an $r$-round quantum protocol which computes APSP.
    For $n\ge5$, we set parameters $s=\lfloor(n-1)/2\rfloor$ and $k=s(s-1)/2$.
    Alice and Bob, respectively receiving $x$ and $y$, compute $\text{INT}_{k,\rho}(x,y)$ by exchanging quantum messages.
    They construct the instance $G=(V_a\uplus V_b,E)$ described above.
    We will show how Alice simulates $\mathcal A$ on the nodes of $V_a$. Bob does the same for $V_b$.
    
    At the $t^\text{th}$ round, for each message sent to a node in $V_b$ by some node in $V_a$ while executing the $t^\text{th}$ round of $\mathcal A$, Alice sends the same message to Bob along with the ID of the receiver.
    For each message received from Bob, Alice ``virtually" allocate it to the corresponding node.
    Communication inside $V_a$ and local computation on nodes of $V_a$ can be done without communicating with Bob.
    
    After simulating $\mathcal A$, Alice knows $d(a_i,b_j)$ for $i,j\in[s]$.
    She counts the number of $(i,j)$ such that $1\le i<j\le s$ and $d(a_i,b_j)=3$ to obtain $|x\cap y|$ and decide $\text{INT}_{k,\rho}(x,y)$ for any $\rho$.
    The quantum communication complexity between Alice and Bob is basically the same with the number of qubits passing over edges between $V_a$ and $V_b$ while simulating $\mathcal A$.
    The set of edges between $V_a$ and $V_b$ is $\{(a_i,b_i):i\in[s]\}$, which is of size $s$.
    Recall that $\mathcal A$ is of $r$ rounds and edges have bandwidth $O(\log n)$ qubits.
    For any $\rho$, the quantum communication complexity of $\text{INT}_{k,\rho}$ is $O(rs\cdot\log n)$. By setting $\rho=\lfloor k/2\rfloor$, Lemma~\ref{int_bound} gives a lower bound $\Omega(k)$.
    Therefore, $r=\Omega(\frac k{s\cdot\log n})=\Omega(\frac n{\log n})$ as $s=\Theta(n)$ and $k=\Theta(n^2)$.
\end{proof}


\subsection{Eccentricities}

To prove the quantum lower bounds of computing the exact and approximate eccentricities shown in Theorem~\ref{ecc_bound} and Theorem~\ref{ecc_approx_bound}, we again establish a reduction from the intersection function, and construct an underlying network graph following the idea of \cite{AbboudCK16}.

\subsubsection{Exact eccentricities}

Let $n\ge23$ be an integer, we construct an instance $G=(V,E)$ of $n$ nodes.
In the first step of construction in Figure~\ref{fig:construction}, we set $k$ to be the maximum integer satisfying $3k+4(\lfloor\log_2(k-1)\rfloor+1)+6\le n$, and set $s=\lfloor\log_2(k-1)\rfloor+1$, $n'=3k+4s+6$, $V=V_a\uplus V_b$ where
\begin{align*}
V_a= & \bigcup_{p\in[k]}\{a_p\}\cup\bigcup_{i\in[s]}\{a^0_i,a^1_i\}\cup\{a'_1,a'_2,a'_3\}\cup\bigcup_{i\in[n-n']}\{a''_i\}, \\
V_b= & \bigcup_{p\in[k]}\{b_p\}\cup\bigcup_{i\in[s]}\{b^0_i,b^1_i\}\cup\{b'_1,b'_2,b'_3\}\cup\bigcup_{p\in[k]}\{b''_p\}.
\end{align*}
We include some additional nodes $a''_1,\cdots,a''_{n-n'}$ in $V$ so that $|V|=n$.

For convenience, let $B(p,i)$ denote the $i^\text{th}$ bit in binary expression of integer $p-1$.
We describe the edges independent of $x$ and $y$.
For each $p\in[k]$, edges between $a_p$ and nodes in $\biguplus_{i\in[s]}\{a^0_i,a^1_i\}$ are added according to binary expression of integer $p-1$, i.e., $a_p$ is connected to $a^{B(p,i)}_i$ for each $i\in[s]$.
Moreover, $a_p$ and $a'_1$, $a'_1$ and $a'_2$, $a'_2$ and $a'_3$ are all connected.
Edges between nodes of $V_b$ are linked in a similar way except that $b_p$ is connected to $b''_p$ for each $p\in[k]$.
In addition, $a'_1$ is connected to $a''_i$ for each $i\in[n-n']$.
As for edges between $V_a$ and $V_b$, an edge linking $a'_3$ and $b'_3$ are added.
$a^0_i$ and $b^1_i$, $a^1_i$ and $b^0_i$ are connected for each $i\in[s]$.

Given inputs $x$ and $y$, Alice and Bob construct the rest of edges.
For $p\in[k]$, Alice adds an edge $(a_p,a'_3)$ if $x_p=1$, and Bob adds an edge $(b_p,b'_3)$ if $y_p=1$.

Formally, we have $E=E_a\uplus E_b\uplus E_c$ where
\begin{align*}
    E_a= & \bigcup_{p\in[k]:x_p=1}\{(a_p,a'_3)\}
           \cup\bigcup_{p\in[k],i\in[s]}\{(a_p,a^{B(p,i)}_i)\} \\
         & \cup\bigcup_{p\in[k]}\{(a_p,a'_1)\}
           \cup\{(a'_1,a'_2),(a'_2,a'_3)\}
           \cup \bigcup_{i\in[n-n']}\{(a'_1,a''_i)\}, \\
    E_b= & \bigcup_{p\in[k]:y_p=1}\{(b_p,b'_3)\}
           \cup\bigcup_{p\in[k],i\in[s]}\{(b_p,b^{B(p,i)}_i)\} \\
         & \cup\bigcup_{p\in[k]}\{(b_p,b'_1)\}
           \cup\{(b'_1,b'_2),(b'_2,b'_3)\}
           \cup\bigcup_{p\in[k]}\{(b_p,b''_p)\}, \\
    E_c= & \bigcup_{i\in[s]}\{(a^0_i,b^1_i),(a^1_i,b^0_i)\}
           \cup\{(a'_3,b'_3)\}.
\end{align*}
See also Figure~\ref{ecc_graph} for better understanding.
The set of red lines is an example of edges added according to $x$ and $y$.
\begin{figure}[ht]
	\centering
	\includegraphics[width=0.5\linewidth]{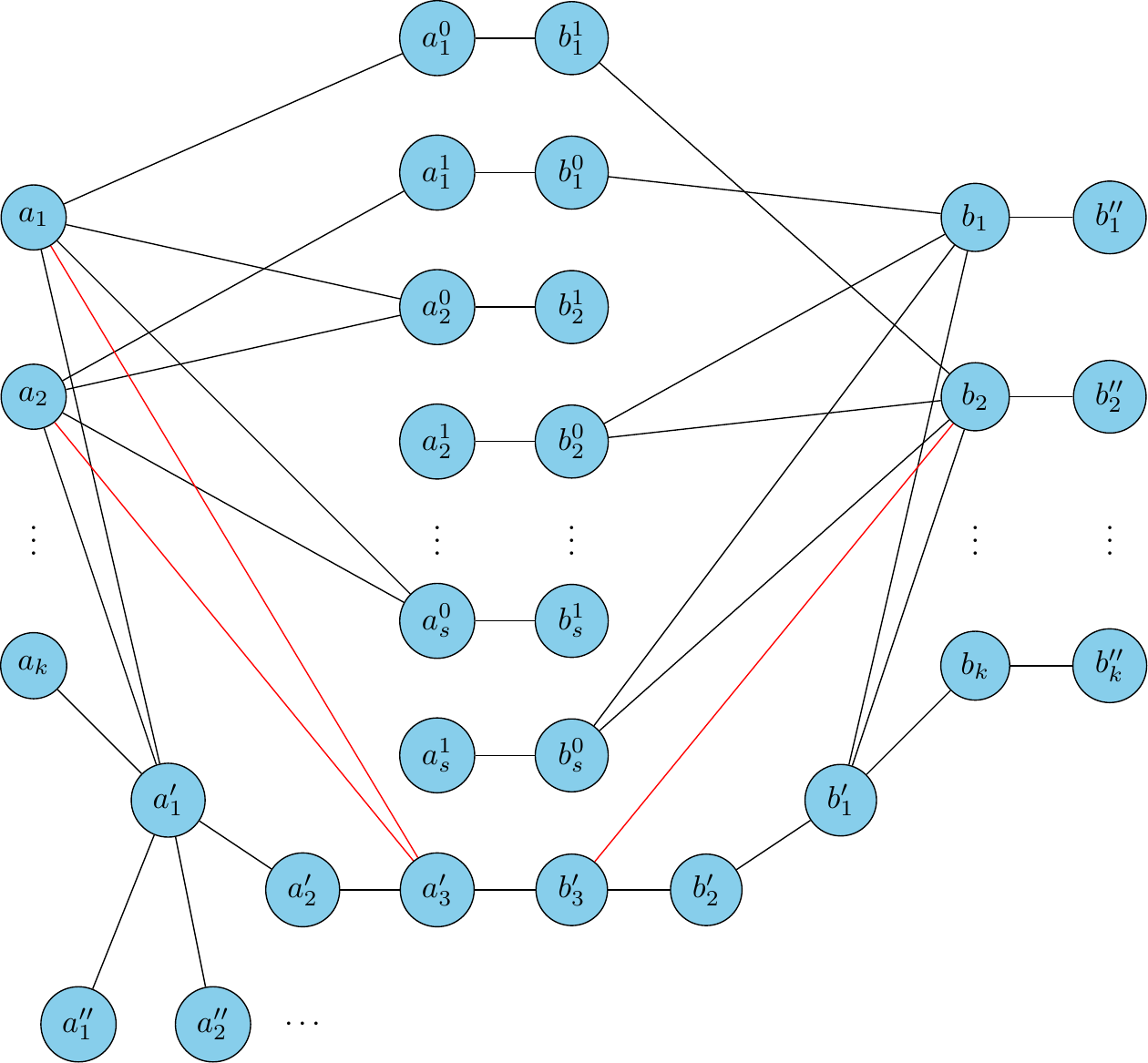}
	\caption{Reduction graph (eccentricity)}
	\label{ecc_graph}
\end{figure}

\begin{proposition}
    \label{ecc_prop1}
    For $p\in[k]$, $d(a_p,b''_p)\le6$ and $d(a_p,v)\le5$ for $v\in[V]\setminus\{b''_p\}$.
\end{proposition}
\begin{proof}
    There are paths $(a_p\to a'_1\to a'_2\to a'_3\to b'_3\to b'_2)$, $(a_p\to a'_1\to a''_i)$ for each $i\in[n-n']$ and $(a_p\to a'_1\to a_q)$ for each $q\in[k]\setminus\{p\}$. Hence,
    $$\forall\ v\in\bigcup_{q\in[k]}\{a_q\}\cup\{a'_1,a'_2,a'_3,b'_3,b'_2\}\cup\bigcup_{i\in[n-n']}\{a''_i\},\ d(a_p,v)\le5.$$
    Since $s=\lfloor\log_2(k-1)\rfloor+1$, for $i\in[s]$, one can always find integers $x,y\in\{0,\cdots,k-1\}$ such that the $i^\text{th}$ bit in binary expression of $x$ ($y$) is $0$ ($1$). And thus indices $q^0=x+1\in[k]$, $q^1=y+1\in[k]$ satisfying $B(q^0,i)=0$ and $B(q^1,i)=1$.
    There are paths $(a_p\to a'_1\to a_{q^0}\to a^0_i\to b^1_i)$ and $(a_p\to a'_1\to a_{q^1}\to a^1_i\to b^0_i)$. Hence,
    $$\forall\ v\in\bigcup_{i\in[s]}\{a^0_i,a^1_i,b^0_i,b^1_i\},\ d(a_p,v)\le4.$$
    For $q\in[k]\setminus\{p\}$, one can always find an index $i\in[s]$ such that $B(p,i)\ne B(q,i)$.
    There is a path $(a_p\to a^{B(p,i)}_i\to b^{B(q,i)}_i\to b_q\to b'_1\to b_p\to b''_p)$ and thus
    $$d(a_p,b_q)\le3,\ d(a_p,b'_1)\le4,\ d(a_p,b_p)\le5,\ d(a_p,b''_p)\le6.$$
    Since $b''_q$ is a neighbour of $b_q$, $d(a_p,b''_q)\le4$.
\end{proof}

\begin{proposition}
    \label{ecc_prop2}
    For $p\in[k]$, $d(a_p,b''_p)<6$ if and only if $x_p=y_p=1$.
\end{proposition}
\begin{proof}
    It is equivalent to prove $d(a_p,b_p)<5$ if and only if $x_p=y_p=1$ because $b_p$ is the only neighbour of $b''_p$ and $d(a_p,b''_p)=d(a_p,b_p)+1$.
    If $x_p=y_p=1$, there is a path $(a_p\to a'_3\to b'_3\to b_p)$ and thus $d(a_p,b_p)\le3$.
    We now only need to prove the necessity.
    
    For $q\in[k]\setminus\{p\}$, obviously $d(a_p,a_q)=d(b_p,b_q)=2$.
    Furthermore, $d(a_p,b_q)=3$ because there is a path of length $3$ between $a_p$ and $b_q$, and they share no neighbour.
    $d(a_q,b_p)=3$ by symmetry.
    
    We use the graphical notation $u\leadsto v$ to denote any shortest path from $u$ to $v$.
    If edge $(a'_3,b'_3)$ is excluded from $E$, the shortest path between $a_p$ and $b_p$ will be of the form $(a_p\leadsto a_q\leadsto b_p)$ or $(a_p\leadsto b_q\leadsto b_p)$ for some $q\in[k]\setminus\{p\}$. Such paths are of length $5$.
    
    Now we consider paths between $a_p$ and $b_p$ where edge $(a'_3,b'_3)$ is involved.
    Such paths are of the form $(a_p\leadsto a'_3\to b'_3\leadsto b_p)$.
    We use the following fact:
    $d(a_p,a'_3)=3$ if $x_p=0$ and $d(a_p,a'_3)=1$ if $x_p=1$;
    $d(b_p,b'_3)=3$ if $y_p=0$ and $d(b_p,b'_3)=1$ if $y_p=1$.
    Thus, such paths have length less than $5$ if and only if $x_p=y_p=1$.
\end{proof}

A corollary of Proposition~\ref{ecc_prop1} and Proposition~\ref{ecc_prop2} tells that for any $p\in[k]$, $a_p$ has eccentricity less than $6$, i.e., $e(a_p)<6$, if $x_p=y_p=1$, and $e(a_p)=6$ otherwise.
So the number of $p\in[k]$ such that $e(a_p)<6$ is the same with the number of $p\in[k]$ with $x_p=y_p=1$.
And for any $\rho$, $\text{INT}_{k,\rho}(x,y)$ can be inferred by computing eccentricities.
We are ready to prove the quantum lower bound of computing exact eccentricities.

\begin{proof}[Proof of Theorem~\ref{ecc_bound}]
    Let $\mathcal A$ be any $r$-round quantum protocol which computes the exact eccentricity for each node.
    By the same argument mentioned in the proof of Theorem \ref{apsp_bound}, Alice and Bob simulate $\mathcal A$ on the instance $G=(V_a\uplus V_b, E_a\uplus E_b\uplus E_c)$ described above.
    After that, Alice knows $e(a_p)$ for $p\in[k]$.
    By counting the number of $p\in[k]$ satisfying $e(a_p)<6$, she can obtain $|x\cap y|$.
    Recall that $k=\Theta(n)$ and the set of edges between nodes simulated by Alice and those simulated by Bob, which is $E_c$, is of size $\Theta(s)=\Theta(\log k)=\Theta(\log n)$.
    The quantum communication complexity of $\text{INT}_{k,\rho}$ is $O(r\cdot\log^2n)$ for any $\rho$.
    Combining Lemma~\ref{int_bound}, $r=\Omega(\frac k{\log^2n})=\Omega(\frac n{\log^2n})$.
\end{proof}

\subsubsection{Approximate eccentricities}

The above argument can be extended to prove a more general quantum lower bound of approximating eccentricities by introducing a parameter $\ell$ and replacing some edges of instance $G=(V,E)$ with paths of length $\ell$.
More specifically, except for edges between $V_a$ and $V_b$, i.e., edges in $E_c$, and edges between $a'_1$ and additional nodes $a''_1,\cdots,a''_{n-n'}$, other edges (including the edges added according to $x$ and $y$) are all replaced with paths containing $\ell-1$ intermediate nodes.
There are $3k+2ks+4+|x|+|y|$ edges needed to be replaced.
Since the replacement may change the number of nodes, we redefine some parameters.

Let $n\ge31\ell-8$ be an integer where $\ell\ge1$ is an undetermined parameter.
We set $k$ to be the maximum integer satisfying
\begin{equation}\label{eqn:k}
  3k + 4(\lfloor\log_2(k-1)\rfloor+1) + 6
+   (\ell-1)(3k + 2k(\lfloor\log_2(k-1)\rfloor+1) + 4 + 2k)
\le n,    
\end{equation}
and set $s=\lfloor\log_2(k-1)\rfloor+1$, $n'=3k+4s+6+(\ell-1)(3k+2ks+4+|x|+|y|)$.
We again include $n-n'$ additional nodes to make sure $|V|=n$.
The instance $G$ is constructed in the same way except that intermediate nodes are added.

By following the argument in the proof of Proposition~\ref{ecc_prop1}, for $p\in[k]$, we have:
$$
\begin{array}{ll}
    \forall\ q\in[k], & d(a_p,a_q)\le2\ell; \\
    \forall\ v\in\{a'_1,a'_2,a'_3,b'_3\}, & d(a_p,v)\le3\ell+1; \\
    \forall\ i\in[n-n'], & d(a_p,a''_i)\le \ell+1; \\
    \forall\ i\in[s], v\in\{a^0_i,a^1_i,b^0_i,b^1_i\}, & d(a_p,v)\le3\ell+1; \\
    \forall\ q\in[k]\setminus\{p\}, & d(a_p,b_q)\le2\ell+1; \\
    \forall\ q\in[k]\setminus\{p\}, & d(a_p,b''_q)\le3\ell+1; \\
    & d(a_p,b'_1)\le3\ell+1.
\end{array}
$$
Moreover, $d(a_p,b'_2),d(a_p,b_p)\le d(a_p,b'_1)+1\le4\ell+1$ and $d(a_p,b''_p)=d(a_p,b_p)+1\le5\ell+1$.

By following the argument in proving Proposition~\ref{ecc_prop2}, we have:
$$
\begin{array}{ll}
    \text{if }x_p=y_p=1, & d(a_p,b'_2)=d(a_p,b_p)=2\ell+1,\ d(a_p,b''_p)=3\ell+1; \\
    \text{otherwise} & d(a_p,b''_p)=5\ell+1.
\end{array}
$$
Thus, $e(a_p)=3\ell+1$ if $x_p=y_p=1$ and $e(a_p)=5\ell+1$ otherwise.
We omit intermediate nodes because the eccentricity of $a_p$ cannot be achieved by any intermediate nodes. 
The ratio of the eccentricity of $a_p$ in the case $x_p=0$ or $y_p=0$ and the one in the case $x_p=y_p=1$ is $\frac{5\ell+1}{3\ell+1}=\frac{5}{3}-\varepsilon$ when $\ell$ is sufficiently large.  Thus, an algorithm of eccentricities  with an approximation ratio better than $\frac{5}{3}$ will induce a quantum communication protocol for the intersection function with inputs $x$ and $y$.

\begin{proof}[Proof of Theorem~\ref{ecc_approx_bound}]
    For any constant $0<\varepsilon<\frac23$ set $\ell=\lceil \frac{2}{9\varepsilon}\rceil$, let $\mathcal A$ be any $r$-round quantum protocol which computes a $(\frac53-\varepsilon)$-approximation of the eccentricities.
    After simulating $\mathcal A$ on the instance $G$ described above, Alice knows, for each $p\in[k]$, a $(\frac53-\varepsilon)$-approximation of the eccentricity of $a_p$, which is denoted by $\hat e(a_p)$.
    As $e(a_p)=3\ell+1$ if $x_p=y_p=1$ and $e(a_p)=5\ell+1$ otherwise, we have:
    $$
    \begin{array}{ll}
        \text{if }x_p=y_p=1, & \frac1{5/3-\varepsilon}\cdot(3\ell+1)\le\hat e(a_p)\le3\ell+1; \\
        \text{otherwise} & \frac1{5/3-\varepsilon}\cdot(5\ell+1)\le\hat e(a_p)\le5\ell+1.
    \end{array}
    $$
    Note that $3\ell+1<\frac1{5/3-\varepsilon}\cdot(5\ell+1)$ for the choice of $\ell$. Alice and Bob are able to output $\abs{x\cap y}$. Thus we obtain a quantum communication protocol for $\text{INT}_{k,\rho}$ with communication cost $O(rs\cdot\log n)$  for any $\rho>0$.
    Combining Lemma~\ref{int_bound}, $r=\Omega(\frac k{s\cdot\log n})=\Omega(\frac k{\log k\cdot\log n})$.
    From Eq.~\eqref{eqn:k} $k\geq \frac{c n}{\log n}$ for some absolute constant $c$. We conclude that $r=\Omega(\frac n{\log^3n})$.
\end{proof}

\section{Summary and Open Problems}

In this paper, we give nearly tight bounds on the complexity of the eccentricities and all pairs shortest paths in the quantum CONGEST model and prove that there is no quantum speedup for these two problems.  In contrast, the computation of diameter and radius have lower round complexity in the quantum CONGEST model than classical CONGEST model when the diameter is small~\cite{GallM18}. It is interesting that quantum communication plays different role for different distance parameters.

For now we only consider the round complexity of the distance parameters of the unweighted network in the quantum CONGEST models. Whether there is quantum speedup for diameter computation in weighted network is still an open problem.

\section*{Acknowledgment}
This work is supported by the National Key R\&D Program of China 2018YFB1003202, National Natural Science Foundation of China (Grant No. 61972191), the Program for Innovative Talents and Entrepreneur in Jiangsu and Anhui Initiative in Quantum Information Technologies Grant No. AHY150100.

\newpage
\bibliographystyle{alpha}
\bibliography{main}

\end{document}